\documentclass[aip,jcp,reprint,a4paper]{revtex4-1}
\usepackage[latin1]{inputenc}
\usepackage{amssymb}
\usepackage{amsmath}
\usepackage{euscript}
\usepackage{mathtools}
\usepackage{color}

\newtheorem{theorem}{Theorem}
\newtheorem{lemma}{Lemma}
\newtheorem{definition}{Definition}
\newenvironment{proof}[1][Proof]
    {\begin{trivlist}\item[\hskip \labelsep] \textit{#1. }}
    {$\square$ \end{trivlist}}

\newcommand{\R}{\mathbb{R}}
\newcommand{\C}{\mathbb{C}}
\newcommand{\N}{\mathbb{N}}
\renewcommand{\d}{\,\mathrm{d}} 
\renewcommand{\i}{\mathrm{i}}
\renewcommand{\H}{\mathcal{H}}
\newcommand{\Cont}{\mathcal{C}}
\newcommand{\X}{\EuScript{X}}
\newcommand{\V}{\EuScript{V}}
\newcommand{\U}{\EuScript{U}}
\newcommand{\lilo}{o}
\newcommand{\id}{\mathrm{id}}
\newcommand{\mtext}[1]{\quad\mathrm{#1}\quad}
\newcommand{\q}[1]{``#1''}
\DeclareRobustCommand{\onehalf}{\textstyle \frac{1}{2}}

\makeatletter
\renewcommand*{\@opargbegintheorem}[3]{\trivlist
  \item[\hskip \labelsep{\bfseries #1\ #2}] \textbf{(#3)}\ \itshape}
\newsavebox{\@brx}
\newcommand{\llangle}[1][]{\savebox{\@brx}{\(\m@th{#1\langle}\)}\mathopen{\copy\@brx\kern-0.5\wd\@brx\usebox{\@brx}}}
\newcommand{\rrangle}[1][]{\savebox{\@brx}{\(\m@th{#1\rangle}\)}\mathclose{\copy\@brx\kern-0.5\wd\@brx\usebox{\@brx}}}
\makeatother

\def\mathclap#1{\mathrm{\hbox to 0pt{\hss$\mathsurround=0pt#1$\hss}}}

\begin{document}
\title{Functional differentiability in time-dependent quantum mechanics}
\author{Markus Penz}
\email{markus.penz@uibk.ac.at}
\author{Michael Ruggenthaler}
\email{michael.ruggenthaler@uibk.ac.at}
\affiliation{Institut für Theoretische Physik, Universität Innsbruck, 6020 Innsbruck, Austria}

\begin{abstract}
In this work we investigate the functional differentiability of the time-dependent many-body wave function and of derived quantities with respect to time-dependent potentials. For properly chosen Banach spaces of potentials and wave functions Fréchet differentiability is proven. From this follows an estimate for the difference of two solutions to the time-dependent Schr\"odinger equation that evolve under the influence of different potentials. Such results can be applied directly to the one-particle density and to bounded operators, and present a rigorous formulation of non-equilibrium linear-response theory where the usual Lehmann representation of the linear-response kernel is not valid. Further, the Fréchet differentiability of the wave function provides a new route towards proving basic properties of time-dependent density-functional theory.
\end{abstract}

\maketitle

\section{Introduction}

Time-dependent quantum mechanics allows us to predict the dynamics of multi-electron systems, such as molecules and atoms. By solving the time-dependent Schr\"odinger equation we determine the interacting many-electron wave function, which in turn determines all physical observables. Since a straightforward solution of the interacting Schr\"odinger equation for realistic quantum systems is usually not feasible numerically, one is bound to perform approximate calculations in computa\-tional chemistry. A prerequisite that such approximation schemes allow reasonable predictions is that already the dynamics of the original fully interacting multi-electron problem are not too sensitive to slight differences in the external field or inter\-action. For instance, the dynamics of a molecule subject to an external laser pulse is usually determined by adopting a dipole approximation instead of a coupling to the full vector potential of the laser pulse. In other words, small changes or inaccuracies in the external potential that drive a quantum system should not have huge effects on the dynamics of that system. It would thus be desirable to give an estimate on how much two solutions differ with respect to the difference of their respective external potentials and interactions. An example for such different interaction terms stems from QED. In Coulomb gauge one can see that any interaction between particles due to a longitudinal charge current is described by the Coulomb interaction, while any interaction due to a transversal current is mediated via the (transversal) photons \cite{greiner}. Hence the interaction between the electrons differs whether we have a transversal current in the quantum system or not \cite{breit}. Provided we can show that the solution to the time-dependent Schr\"odinger equation has a smooth linear response due to variations in the external potential or interaction, which in the framework of Banach spaces presented here is a so called Fréchet derivative, we can derive such an estimate with the help of the fundamental theorem of calculus on Banach spaces. Hence we are interested in the functional differentiability of the time-dependent wave function with respect to time-dependent perturbations, i.e., non-equilibrium linear-response theory. We point out that we consider explicitly time-dependent systems, such that the usual linear-response formalism in terms of the Lehmann representation \cite{GF, stefanucci-vanleeuwen} is in general not applicable. This has to do with the fact that the evolution operator of an explicitly time-dependent Hamiltonian cannot be expressed by a simple spectral representation.

Also for approximation schemes used in quantum mechanics, functional differentiability is of central importance. For instance, in density-functional theory (DFT) \cite{DFT, DFT2} and the associated Kohn-Sham construction, functional differentiability of certain energy functionals with respect to the one-particle density is vital \cite{lieb, chayes, eschrig, lammert1, lammert2, kvaal}. Functional differentiability in the time-dependent extension of DFT, so-called time-dependent density-functional theory (TDDFT) \cite{TDDFT, TDDFT2}, is employed to extend the validity of the theory beyond its usual standard formulation in terms of Taylor-expandable external potentials \cite{runge} and to avoid the very strong restriction of Taylor-expandable densities in the construction of a Kohn-Sham system \cite{vanleeuwen, review}. This is achieved by using density-response theory \cite{vanleeuwen2} or  directly the functional differentiability of an observable that describes the internal-force density \cite{ruggenthaler1, ruggenthaler2}. In the first case the considerations are restricted to the perturbation of time-independent problems. To use a similar approach for explicitly time-dependent systems one needs a well-defined non-equilibrium density-response theory. In the second approach Fréchet differentiability of the internal-force density is already assumed and still needs to be justified mathematically. Therefore, functional differentiability of the wave function is of central importance also within TDDFT. \\ 

In this work we will present such rigorous results on the functional differentiability of solutions to the time-dependent Schr\"odinger equation and provide similar results for the density and bounded operators. Since the Schr\"odinger Hamiltonian does not depend on spin we disregard the spin degrees in the many-electron wave function in this work. How the spin-dependent wave function can be constructed from the spin-independent ones is discussed in detail in the books of McWeeny\cite{mcweeny} and Wilson\cite{wilson}. We therefore consider the Cauchy problem for the time-dependent Schrödinger equation (TDSE) on the Hilbert space $\H = L^2(\R^n, \C)$ of quantum states equipped with  the scalar product $\langle \cdot,\cdot \rangle$. The dimensionality of the configuration space is usually $n=3N$ with $N$ the number of particles and we only consider a finite time interval $[0,T]$ and a time-dependent, possibly unbounded scalar potential $v : [0,T] \times \R^n \rightarrow \R$. Note that $v$ can contain (possibly time-dependent $N$-body) interaction terms too. The Hamiltonian governing the dynamics of the system is then $H[v] = H_0 + v = -\Delta + v$, where $\Delta$ is the $n$-dimensional Laplacian, and the Cauchy problem is stated as follows.
\begin{equation}\label{cauchy-problem}
\left.\begin{array}{l}
    \i \partial_t \psi = H[v]\, \psi \\
    \psi(0) = \psi_0 \in \H
\end{array}\right\}
\end{equation}

However, if we want to find a unique solution to the above Cauchy problem, not every wave function is a valid initial state, since the Schr\"odinger Hamiltonian $H[v]$ is an unbounded, self-adjoint operator on $\H$ (and hence is not defined on all of $\H$ by the Hellinger-Toeplitz theorem). To overcome this restriction and to allow for every $\psi_0 \in \H$ we derive a generalized version of \eqref{cauchy-problem}. We introduce the $H_0$-interaction picture as the unitary transformation $\psi(t) = U_0(t)\hat\psi(t)$ with the one-parameter unitary group $U_0(t) = \exp(-\i t H_0) = \exp(\i t \Delta)$. Put into \eqref{cauchy-problem} and integrated over time this leaves us with an integral equation of Volterra type that when transformed back we call the ``mild'' Schrödinger equation.
\begin{align}\label{schroequ-mild}
\begin{split}
\psi([v],t) &= U([v],t,0)\psi_0 \\
&= U_0(t)\psi_0 - \i\int_0^t U_0(t-s) v(s) \psi([v],s) \d s
\end{split}
\end{align}
The unique solution to this problem for given $v$ and $\psi_0$, if existent, we call the quantum trajectory $\psi[v] : [0,T] \times \R^n \rightarrow \C$ in the Schrödinger picture. We follow closely the work of Yajima \cite{yajima} for existence and uniqueness of such ``mild'' solutions to \eqref{schroequ-mild} considering a Banach space $\V$ of possible potentials. We then consider how a slight variations of $v$ in the direction $w$ would alter this result, i.e. forming the directional derivative
\begin{equation}\label{gateaux-deriv}
\delta\psi[v;w] = \lim_{\lambda \rightarrow 0} \frac{1}{\lambda} \big(\psi[v+\lambda w]-\psi[v]\big)
\end{equation}
where the limit is taken with respect to the norm of an appropriate Banach space $\X$ of quantum trajectories. If this yields a linear, bounded map $\delta\psi[v;\cdot] \in \mathcal{B}(\V, \X)$ it is called the Gâteaux derivative at point $v$. If it is further defined at all points $v$ of an open subset $\U \subset \V$ and $v \mapsto \delta\psi[v;\cdot]$ is a continuous mapping $\U \rightarrow \mathcal{B}(\V,\X)$ then the Gâteaux derivative equals the Fréchet derivative and we write $\psi \in \Cont^1(\U,\X)$. In the Fréchet case the limit in \eqref{gateaux-deriv} holds uniformly for every kind of path within $\U$ towards zero.\\

All the quantum trajectories we consider in this work are within the set $\Cont^0([0,T],\H)$ of continuous maps from the time interval to the Hilbert space of quantum states. It will be equipped with the norm $\|\varphi\|_{2,\infty} = \sup_{t\in[0,T]}\|\varphi(t)\|_2$ to make it a Banach space (it is closed because the supremum norm implies uniform convergence on the compact time interval). The main result of this work is now to show that the definitions of $\V$ and $\X$ following Yajima \cite{yajima} are sufficient to make $\psi[v]$ Fréchet differentiable. To show this we have to carefully investigate Schrödinger dynamics subject to time-dependent potentials. To this end and in order to be self-contained we first start by revisiting the mathematical approach to the TDSE developed by Yajima. The important special case of singular Coulombic potentials is discussed in a separate section. After the proof of the main theorem we present its application to linear response theory and TDDFT by deriving estimates for the variation of expectation values of bounded self-adjoint operators and of the one-particle density. While the first result leads to the well-known non-equilibrium version of Kubo's formula, the second application puts non-equilibrium density-response theory \cite{stefanucci-vanleeuwen} on rigorous grounds. Finally we discuss implications of our results for TDDFT.

\section{Review of Schrödinger dynamics with time-dependent potentials}

We start by giving a review of uniqueness and existence results for the TDSE based on the work of Yajima \cite{yajima}. To this end we first define Lebesgue-spaces $L^{q,\theta}$ over $[0,T] \times \R^n$ as the set of functions with finite norm
\[
\|\varphi\|_{q,\theta} = \left( \int_0^T \left( \int_{\R^n} |\varphi(t,x)|^q \d x \right)^{\theta/q} \d t \right)^{1/\theta} < \infty
\]
modulo the null set $\{ \varphi : \|\varphi\|_{q,\theta}=0 \}$. The first superscript $q$ denotes the $L^q$ space in spatial coordinates and $\theta$ the $L^\theta$ space over the (finite) time interval. Latin characters are always used for the space part and Greek ones for time. The special cases $q$ or $\theta = \infty$ are possible and defined in the usual way with the supremum (uniform) norm in time and the essential supremum norm in space. Such Lebesgue-spaces are the building blocks of the Banach spaces of quantum trajectories and corresponding potentials.

\begin{definition}[Banach spaces of quantum trajectories]\label{def-X}
Let the principal indices for the Banach space $\X$ be $2 \leq q \leq \infty$ and $2 < \theta \leq \infty$ with their dual exponents $q' = q/(q-1), \theta' = \theta/(\theta-1)$ and therefore fulfilling $1 \leq q' \leq 2$ and $1 \leq \theta' < 2$ as well as the typical Hölder relations $1/q + 1/q' = 1$ and $1/\theta + 1/\theta' = 1$. The assumed relation between those indices is $2/\theta = n(1/2 - 1/q)$ which implies $q < 2n/(n-2)$ for $n \geq 3$. We define $\X$ and its topological dual $\X'$ by
\begin{align*}
\X &= \Cont^0([0,T], \H) \cap L^{q,\theta}, \\
\X' &= L^{2,1} + L^{q',\theta'}.
\end{align*}
\end{definition}

The special relation between the exponents $q,\theta$ of this Banach space is called \emph{Schrödinger-admissible}. In the paper by D'Ancona et al.\cite{dancona} the scope is widened slightly to $\theta \geq 2$ with the choice $(n,\theta,q) = (2,2,\infty)$ ruled out. The norms of $\X$ and $\X'$ are
\begin{align*}
\|\varphi\|_\X =& \|\varphi\|_{2,\infty} + \|\varphi\|_{q,\theta}, \\
\|\varphi\|_{\X'} =& \inf\{ \|\varphi_1\|_{2,1} + \|\varphi_2\|_{q',\theta'} : \varphi_1 \in L^{2,1}, \varphi_2 \in L^{q',\theta'}, \\& \varphi=\varphi_1+\varphi_2 \}.
\end{align*}
Those two spaces are related as duals by the space-time scalar product $\llangle\cdot,\cdot\rrangle$ defined by
\[
\llangle\varphi,\psi\rrangle = \int_0^T \int_{\R^n} \overline{\varphi}(t,x) \psi(t,x) \d x \d t.
\]
Firstly, we consider freely evolving trajectories given by application of the free evolution operator belonging to the Cauchy problem \eqref{cauchy-problem} with $v=0$, i.e., a mapping from initial states to trajectories.
\begin{align}\label{free-evolution}
\begin{split}
U_0: \H &\longrightarrow \X \\
\psi_0 &\longmapsto U_0\psi_0 = (t \mapsto U_0(t)\psi_0).
\end{split}
\end{align}
To guarantee them lying within the space $\X$ we rely on an inequality in its original form due to Strichartz \cite{strichartz} considering the wave equation but a version for solutions to the free Schrödinger equation is also available\cite{ginibre-velo,yajima2}.

\begin{theorem}[Strichartz inequality]\label{strichartz-ineq}
Let the exponents $q,\theta$ be like in Definition \ref{def-X} then there exists a constant $C_0$ such that for every $\psi_0 \in \H$ it holds
\[ \| U_0 \psi_0 \|_{q,\theta} \leq C_0 \|\psi_0\|_2. \]
\end{theorem}

Combined with the unitarity of $U_0(t)$ we easily get the desired inequality that shows the stability of the free evolution within $\X$.
\begin{equation}\label{free-evol-inequ}
\| U_0 \psi_0 \|_\X \leq (1+C_0) \|\psi_0\|_2.
\end{equation}
The spaces of trajectories will now be accompanied by the corresponding Banach spaces for potentials that guarantee stability in $\X$ also for non-free evolution operators.

\begin{definition}[Banach spaces of potentials]\label{def-V}
Related to $\X$ we define $\V$ demanding of its indices $p \geq 1,\alpha \geq 1,\beta > 1$ that $0 \leq 1/\alpha < 1-2/\theta$ and $1/p = 1-2/q$.
\[
\V = L^{p,\alpha} + L^{\infty,\beta}
\]
\end{definition}

The condition on $p$ actually guarantees finite potential energy at almost all times for $v(t) \in L^p$ and a state $\psi(t) \in L^q$ thus $n(t) \in L^{q/2}$ because $1/p + 2/q = 1$ means $v(t)n(t) \in L^1$. Note that because of the condition on $q$ in Definition \ref{def-X} this set of inequalities demands $p > \frac{n}{2}$ for $n \geq 3$ therefore demanding $p \rightarrow \infty$ for very large particle numbers which rules out Coulombic singular potentials as will be discussed in Section \ref{sect-coulomb}. Further we note that the Banach space of potentials $\V$ includes potentials depending on more than two particle coordinates and potentials without special symmetry conditions thus acting differently on different particles (destroying any assumed Bose or Fermi symmetry).
\\

The spaces $\X, \X'$, and $\V$ are linked in the following lemma, a slightly generalised form of the Hölder inequality, taken from Yajima \cite[Lemma 2.3]{yajima}.

\begin{lemma}\label{lemma-mult-op}
A multiplication operator $v \in \V$ is a bounded operator $\X \rightarrow \X'$ and fulfils
\[
\|v \varphi\|_{\X'} \leq T^* \|v\|_\V \|\varphi\|_\X
\]
with
\begin{equation}\label{Tstar}
T^* = \max\{ T^{1-1/\beta}, T^{1 - 2/\theta-1/\alpha} \}
\end{equation}
monotonously increasing in $T$.
\end{lemma}

\begin{proof}
We remember the partitioning $v = v_1 + v_2$ with $v_1 \in L^{p,\alpha}, v_2 \in L^{\infty,\beta}$ given by the norm of $\V$ and use Hölder's inequality for each part of $v\varphi = v_1\varphi + v_2\varphi$. To get the final result we need to change the time indices of the norms to bigger values which is possible with the simple relation (for arbitrary $m,\gamma,\rho$ and $\rho>\gamma$ using Hölder):
\begin{align*}
\|f\|_{m,\gamma} &= \|1 \cdot f\|_{m,\gamma} \\
&\leq \|1\|_{\infty,\gamma\rho/(\rho-\gamma)} \|f\|_{m,\rho} \\
&= T^{1/\gamma-1/\rho} \|f\|_{m,\rho}.
\end{align*}
For the $L^{q',\theta'}$-part of $\X$ we have with $1/q'-1/q=1/p$
\begin{align*}
\|v_1\varphi\|_{q',\theta'} &\leq \|v_1\|_{p,\theta\theta'/(\theta-\theta')} \|\varphi\|_{q,\theta} \\
&\leq T^{1-2/\theta-1/\alpha} \|v_1\|_{p,\alpha} \|\varphi\|_{q,\theta}
\end{align*}
and for the $L^{2,1}$-part 
\[
\|v_2\varphi\|_{2,1} \leq \|v_2\|_{\infty,1} \|\varphi\|_{2,\infty} \leq T^{1-1/\beta} \|v_2\|_{\infty,\beta} \|\varphi\|_{2,\infty}.
\]
The right hand side of the Lemma's statement clearly includes those two estimates which concludes the proof.
\end{proof}

We also adopt the definition of the trajectory map $Q$ from Yajima \cite[2.1]{yajima} but add the relevant potential $v$ as an index to the notation.
\begin{align}\label{def-Qv}
\begin{split}
Q_v : \X &\longrightarrow \X \\
\varphi &\longmapsto \left(t \mapsto - \i\int_0^t U_0(t-s) v(s) \varphi(s) \d s \right) 
\end{split}
\end{align}
Using Lemma \ref{lemma-mult-op} and dual-space tricks $Q_v$ is shown in \cite{yajima} to be bounded with operator norm $\|Q_v\| \leq C_Q T^* \|v\|_\V$ with a fixed constant $C_Q>0$. $Q_v\psi[v]$ is just the integral term in \eqref{schroequ-mild} and thus we can write the mild Schrödinger equation briefly as
\begin{equation}\label{schroequ-mild-2}
\psi[v] = U_0\psi_0+Q_v\psi[v].
\end{equation}
Inverting \eqref{schroequ-mild-2} yields a Neumann series which we can write as an equation to determine not $\psi([v],t)$ at a given instant but as a whole trajectory $\psi[v]:t \mapsto \psi([v],t)$ within $\X$.
\begin{equation}\label{neumann-series}
\psi[v] = (\id - Q_v)^{-1} U_0 \psi_0 = \sum_{k=0}^\infty Q_v^k U_0\psi_0
\end{equation}
This series actually converges if $T$ is short enough s.t.~$\|Q_v\| < 1$, which is always possible for fixed $v \in \V$. The uniqueness of solutions to the Schrödinger equation for longer time intervals is still guaranteed by a continuation procedure, taking $\psi([v],T)$ as a new initial value. This result can be used to define an evolution operator by $\psi([v],t) = U([v],t,s)\psi([v],s)$ with start time $s$ and end time $t$ which in Yajima \cite{yajima} is shown to fulfill the usual properties of an evolution system (we refer to the book of Pazy\cite{pazy} for a detailed definition). Note that $U([0],t,s) = U_0(t-s)$ is just the free evolution. Analogously to $U_0$ in \eqref{free-evolution} we define the evolution under a potential $v \in \V$ as a mapping $U[v]$ from initial states to trajectories.
\begin{align*}
U[v] : \H &\longrightarrow \X \\
\psi_0 &\longmapsto U[v]\psi_0 = (t \mapsto U([v],t,0)\psi_0)
\end{align*}

This result shows existence and uniqueness of solutions to the Schrödinger equation with a potential $v \in \V$. A more direct and thus simpler Strichartz-like estimate is due to D'Ancona et al.\cite{dancona}~and uses a fixed-point technique applied to a contraction derived from the implicit form of the mild Schrödinger equation \eqref{schroequ-mild-2}. Because we later refer to the estimate derived in its proof it will be given here.

\begin{theorem}\label{Cv-strichartz}
For arbitrary albeit finite $T>0$ and $v \in \V$ (in certain cases $T \rightarrow \infty$ becomes feasible) the solution to the mild Schrödinger equation yields the Strichartz estimate 
\[
\|\psi[v]\|_{q,\theta} \leq C_v \|\psi_0\|_2
\]
where $C_v=2 M^{1/\theta} (1+C_0)$. For the definition of $M(v) \in \N$ note the details in the beginning of the proof.
\end{theorem}

\begin{proof}
Firstly divide the time interval $[0,T]$ into a finite number $M$ of subintervals $I_1, \ldots, I_M$. Each subinterval be short enough such that $C_Q |I_m|^*\|v\|_{\V | I_m} \leq \onehalf$. Note that by this division also an infinite time interval $[0,\infty)$ gets feasible for potentials decaying fast enough such as in scattering processes. Now take the recursive formula \eqref{schroequ-mild-2} and define a map $\Phi : \X \rightarrow \X$
\[
\Phi(\psi) = U_0 \psi_0 + Q_v \psi.
\]
A fixed point of this map would be a solution to the mild Schrödinger equation. If we limit ourselves to any of the subintervals we have the following inequality by \eqref{free-evol-inequ} and the estimate for the operator norm of $Q_v$.
\[
\|\Phi(\psi)\|_{\X | I_m} \leq (1+C_0) \|\psi_0\|_2 + \onehalf \|\psi\|_{\X | I_m}.
\]
Now $\Phi$ clearly defines a contraction mapping and the unique fixed point $\psi = \Phi(\psi) \in \X|_{I_m}$ fulfils $\|\psi\|_{\X | I_m} \leq 2(1+C_0) \|\psi_0\|_2$. The final step is to concatenate all of these estimates to get one for the full time interval.
\begin{align*}
\|\psi\|_{q,\theta} &= \left( \sum_{m=1}^M \int_{I_m} \|\psi(t)\|_q^\theta \d t \right)^{1/\theta} \\
&\leq \left( \sum_{m=1}^M \|\psi\|_{\X | I_m}^\theta \right)^{1/\theta} \\
&\leq \left( \sum_{m=1}^M (2(1+C_0) \|\psi_0\|_2)^\theta \right)^{1/\theta} \\[0.5em]
&\leq 2 M^{1/\theta} (1+C_0) \|\psi_0\|_2
\end{align*}
\end{proof}

\section{Dynamics with Coulombic potentials}\label{sect-coulomb}

The study of Yajima \cite{yajima} partly revisited above holds for arbitrary spatial dimension $n$ thus in principle allowing multiple particles in three-dimensional space. Contrary to investigations on general evolution equations \cite{phillips,kato1953,pazy} it concentrates on the Schrödinger case and is \q{taking the characteristic features of Schrödinger equations into account [to] establish a theorem [...] for a larger class of potentials than in existing abstract theories.} The most significant such feature is the availability of Strichartz-type estimates. But it is important to mention that the Coulombic case for more than one particle is still ruled out. \\

Take a radial singular potential $v$ with its centre at the origin and $v = v_1 + v_2$ with $v_1 \in L^p, v_2 \in L^\infty$ like demanded in Definition \ref{def-V} for almost all times. We can always assume the support of $v_1$ confined in a ball $r=|x| \leq 1$ because the outer part is bounded and thus in $L^\infty$. The $L^p$ condition now reads in spherical coordinates
\[
\int_0^1 |v_1(r)|^p \, r^{n-1} \d r < \infty.
\]

A singular potential of type $v_1(r) = -r^{-s}$ must therefore fulfil $-p s + n - 1 > -1$ for a converging norm integral which is the same as $s < \frac{n}{p}$ thus $s < 2$ by Yajima's assumption on the potential space for $n \geq 3$. But such a potential is not of Coulombic type if more than one quantum particle in three-dimensional space is considered. Remember that the general form for a centred Coulomb potential for $N$ particles would be
\begin{equation}\label{coulomb-v}
v(x_1,\ldots,x_N) = -\sum_{i=1}^N r_i^{-1}.
\end{equation}
with $r_i = |x_i|$. The $L^p$ condition thus reads for one of the most singular terms
\begin{align*}
\int_{[0,1]^N} r_1^{-p+2} \d r_1 r_2^2 \d r_2 \ldots r_N^2 \d r_N < \infty \\
\mtext{i.e.} \int_0^1 r_1^{-p+2} \d r_1 < \infty,
\end{align*}
and we need $-p+2>-1$. Thus $\frac{n}{2} < p < 3$ which is not feasible for $n\geq 6$, the case of two or more particles. The problem arises even more drastically outside the centre region because of infinitely stretched singularities along all $\{r_i = 0\}$. Also the problem persists for singular interaction terms of the kind $v(x_1,x_2) = |x_1-x_2|^{-1}$ which describe the interaction of charged particles. In this approach this effectively rules out the Coulombic case for systems of more than one particle. 

\section{Fréchet differentiability of the wave function}

To derive an expression for the variational derivative $\delta\psi[v;w]$ consider the two Schrödinger equations
\begin{align}
 \i \partial_t \psi[v] &= H[v] \psi[v], \label{schro-equ-1} \\
 \i \partial_t \psi[v+w] &= H[v + w] \psi[v+w], \label{schro-equ-2}
\end{align}
both with the same initial value $\psi_0$. Assume $U[v]$ to be the evolution system corresponding to the Schrödinger equation \eqref{schroequ-mild} with potential $v$ and unitarily transform $\psi \mapsto \hat\psi$ to the $H[v]$-interaction picture by
\[
\psi(t) = U([v],t,0) \hat{\psi}(t).
\]
Putting this into \eqref{schro-equ-2} we get the Tomonaga-Schwinger equation and its integral, mild version as an analogue of \eqref{schroequ-mild}.
\begin{align}
&\i\partial_t \hat{\psi}[v+w] = \hat{w} \hat{\psi}[v+w] \nonumber\\
&\hat{w}(t) = U([v],0,t) w(t) U([v],t,0) \nonumber\\
\label{schroequ-mild-pertub}
&\hat{\psi}([v+w],t) = \psi_0 - \i\int_0^t \hat{w}(s) \hat{\psi}([v+w],s) \d s
\end{align}
Note that in the case of \eqref{schro-equ-1} with only potential $v$, that is $w=0$, this implies the identity $\hat{\psi}([v],t) = \psi_0$. We thus have from \eqref{schroequ-mild-pertub} and by proceeding recursively
\begin{align}\label{variation-recursion}
\begin{split}
&\hat\psi([v+w],t)-\hat{\psi}([v],t) \\
&= - \i \int_0^t \hat{w}(s) \, \hat\psi([v+w],s) \d s\\
&= - \i \int_0^t \hat{w}(s) \left( \psi_0 - \i \int_0^s \hat{w}(s') \hat\psi([v+w],s') \d s' \right) \!\d s.
\end{split}
\end{align}
With this expression it is easy to take the corresponding Gâteaux limit (shown to converge in the proof of Theorem \ref{Psi-frechet-prepare}) to get a first order approximation.
\begin{align}\label{gateaux-deriv-int}
\begin{split}
\delta \hat\psi([v;w],t) &= \lim_{\lambda \rightarrow 0} \frac{1}{\lambda} \big(\hat\psi([v+\lambda w],t)-\hat{\psi}([v],t)\big) \\
&=- \i \int_0^t \hat{w}(s) \psi_0 \d s
\end{split}
\end{align}
Transformed back to the Schrödinger picture we have
\begin{equation}
\delta\psi([v;w],t) = - \i \int_0^t U([v],t,s) \, w(s) \, \psi([v],s) \d s,
\end{equation}
or as the variation of the evolution operator acting on $\psi_0$ equivalently
\begin{equation}\label{U-gateaux-deriv}
\delta U([v;w],t,0) = - \i \int_0^t U([v],t,s) \, w(s) \, U([v],s,0) \d s.
\end{equation}
This variation within $\V$ is now actually a Fréchet derivative on a bounded set of potentials.

\begin{theorem}\label{Psi-frechet}
For arbitrary albeit finite $T>0$ and initial state $\psi_0 \in \H$ the unique solution to the mild Schrödinger equation is Fréchet-differentiable on $\U \subset \V$ bounded and open, i.e., $\psi \in \Cont^1(\U,\X)$. The following estimate holds for all $v,w \in \V$.
\[
\|\delta\psi[v;w]\|_{2,\infty} \leq (1+C_v)^2 T^* \|w\|_\V \|\psi_0\|_2
\]
\end{theorem}

Here we used the Definitions \ref{def-X} and \ref{def-V} of the associated Banach spaces and the constant $C_v$ from Theorem \ref{Cv-strichartz}. The proof of this theorem is given in detail in the next section. Consequences of this result are discussed in Sections \ref{FrechetBounded} and \ref{FrechetDensity}.

\section{Proof of Theorem \ref{Psi-frechet}}

In the following we first show Fréchet differentiability of $\psi[v]$ for a small time interval. Then we extend the result to arbitrarily large albeit finite times, before we finally deduce the inequality of Theorem \ref{Psi-frechet}.

\begin{theorem}\label{Psi-frechet-prepare}
Let $\psi_0 \in \H$, $\U \subset \V$ bounded and open and $T>0$ short enough such that $C_Q T^* \|v\|_\V < 1$ for all $v \in \U$ then the unique solution to the mild Schrödinger equation is Fréchet-differentiable on $\U$, i.e. $\psi \in \Cont^1(\U,\X)$. Likewise we have the variation of the evolution operator $\delta U : \U \times \V \rightarrow \mathcal{B}(\H, \X)$.
\end{theorem}

\begin{proof}
We use the shorthand notation $R_v = (\id - Q_v)^{-1}$ as this operator is closely related to the resolvent of $Q_v$. Because of the limitation to potentials $v \in \U$ we have convergence of the Neumann series in \eqref{neumann-series} which means boundedness of $R_v$. Due to $Q_{v+w} = Q_v+Q_w$ the resolvent identity
\[
R_{v+w} = R_v (\id + Q_w R_{v+w})
\]
holds. Thus inserting recursively we get from \eqref{neumann-series} the difference
\begin{align*}
\psi[v+\varepsilon w]-\psi[v] &= R_v Q_{\varepsilon w} R_{v+\varepsilon w}U_0\psi_0 \\
&= \sum_{k=1}^\infty (R_v Q_{\varepsilon w})^k R_v U_0\psi_0.
\end{align*}
This series coverges for fixed $v,w$ and small enough $\varepsilon$. We use again linearity $Q_{\varepsilon w} = \varepsilon Q_{w}$ for $\varepsilon \in \R$ and the Gâteaux limit follows immediately.
\begin{align}\label{Psi-frechet-expression}
\begin{split}
\delta \psi[v;w] &= \lim_{\varepsilon \rightarrow 0} \frac{1}{\varepsilon} (\psi[v+\varepsilon w]-\psi[v]) \\
&= R_v Q_w R_v U_0\psi_0
\end{split}
\end{align}
Continuity (and linearity) of the above form of $\delta \psi$ in its second argument is readily established by continuity (and linearity) of $Q_w$ in $w$. This proves Gâteaux differentiability. If we additionally show $v \mapsto \delta\psi[v,\cdot]$ continuous as a mapping $\U \rightarrow \mathcal{B}(\V,\X)$ then a lemma from variational calculus (see for example Lemma 30.4.2 in Blanchard-Brüning \cite{blanchard-bruening}) implies Fréchet differentiability. This is certainly true if $\lim_{h \rightarrow 0}\|\delta\psi[v+h;w] - \delta\psi[v;w]\|_\X = 0$ for all $w \in \V$. We show this by using expression \eqref{Psi-frechet-expression} for $\delta\psi$ and the resolvent identity once more.
\begin{align*}
&\delta\psi[v+h;w] - \delta\psi[v;w] \\[0.7em]
&= (R_{v+h} Q_w R_{v+h} - R_v Q_w R_v) U_0\psi_0 \\
&= \sum_{(j,k) \neq (0,0)}^\infty (R_v Q_h)^j R_v Q_w (R_v Q_h)^k R_v U_0\psi_0
\end{align*}
Again those sums will converge for small enough $h \in \V$ and the expression is well defined. As there is at least one $Q_h$ contained in every term and $\|Q_h\| \leq C_Q T^* \|h\|_\V$ the whole expression goes to $0$ as $h \rightarrow 0$. This makes $\psi : \V \rightarrow \X$ Fréchet differentiable on $\U$.
\end{proof}

Note particularly that if we want to widen the open ball $\U \subset \V$ with radius $R$ of allowed potentials this means the time bound $T$ limited by $T^* < (C_Q R)^{-1}$ gets smaller and vice versa. By dividing the time interval in sufficiently short subintervals with individual evolution operators we can circumvent this limitation as shown by the following proof of the first part of the main theorem.

\begin{proof}[Proof of Theorem \ref{Psi-frechet}, Fréchet differentiability]
We use the way $U[v]$ can be put together by expressions like in \eqref{neumann-series}, each one for a short enough time interval such that convergence is guaranteed. This means take $M \in \N$ large enough and define $\tau = T/M$ and $\tau^*$ like in \eqref{Tstar} such that $C_Q \tau^* \|v\|_\V < 1$ for all $v \in \U$ which is possible due to boundedness of $\U$. We thus have a partition into subintervals $I_1 = [0,\tau], I_2 = [\tau,2\tau], \ldots, I_M=[(M-1)\tau,T]$. Imagine for the time being $M=2$ is large enough, later we generalise this case. Now we have
\[
(U[v]\psi_0)(t) = \left\{
\begin{array}{ll}
 U([v], t, 0) \psi_0 & \mbox{for}\; t \in I_1 \\
 U([v], t, \tau) \, U([v], \tau, 0) \psi_0 & \mbox{for}\; t \in I_2.
\end{array} \right.
\]
The variations of the individual evolution operators are well defined, proven in Theorem \ref{Psi-frechet-prepare}, one just needs to shift the potentials accordingly in time to have the $Q_v$ and $Q_w$ operators acting correctly as the integrals therein always start at $t=0$. To determine $\delta U[v;w]$ we put in the expansion $U[v+w] \in U[v] + \delta U[v;w] + \lilo(\|w\|_\V)$ as $w \rightarrow 0$ for all evolutions.
\begin{align*}
(U[v+w]\psi_0)(t) \in \left\{
\begin{array}{ll}
 U([v], t, 0) \psi_0 + \delta U([v;w], t, 0) \psi_0 \\
 + \lilo(\|w\|_\V) \\
 \multicolumn{1}{r}{\mbox{for}\; t \in I_1} \\[0.5em]
  U([v], t, \tau) \, U([v], \tau, 0) \psi_0 \\
+ \, \delta U([v;w], t, \tau) \, U([v], \tau, 0) \psi_0 \\
+ \, U([v], t, \tau) \, \delta U([v;w], \tau, 0) \psi_0 \\
+ \, \delta U([v;w], t, \tau) \, \delta U([v;w], \tau, 0) \psi_0 \\
+ \, \lilo(\|w\|_\V) \\
 \multicolumn{1}{r}{\mbox{for}\; t \in I_2}
\end{array} \right.
\end{align*}
The quadratic $\delta U$ term is of order $\lilo(\|w\|_\V)$ as $w \rightarrow 0$ as well and can therefore be neglected in the whole $\delta U[v;w]$ expression. We show this with the boundedness of $\delta U$ in its second argument from Theorem \ref{Psi-frechet-prepare}, introducing a bound $C>0$. Further we employ the obvious estimate $\|\varphi(t)\|_2 \leq \|\varphi\|_{\X|I_m} \leq \|\varphi\|_\X$ for $t \in I_m$.
\begin{align*}
&\|\delta U([v;w], \cdot, \tau) \, \delta U([v;w], \tau, 0) \psi_0\|_{\X | I_2} \\
&\leq C \|w\|_\V \|\delta U([v;w], \tau, 0) \psi_0\|_2 \\
&\leq C  \|w\|_\V \|\delta U[v;w] \psi_0\|_{\X|I_1} \\
&\leq C^2 \|w\|_\V^2 \|\psi_0\|_2
\end{align*}
The extension to $M>2$ is straightforward and gives us the following product rule for $\delta \psi[v;w]$ at time $t \in I_m$.
\begin{align*}
&\delta \psi([v;w],t) = (\delta U[v;w]\psi_0)(t) \\
&= \delta U([v;w], t, (m-1)\tau) \ldots U([v], 2\tau,\tau) \, U([v], \tau,0) \psi_0 \\
&+ \cdots \\
&+ U([v], t, (m-1)\tau) \ldots \delta U([v;w], 2\tau,\tau) \, U([v], \tau,0) \psi_0 \\
&+ U([v], t, (m-1)\tau) \ldots U([v], 2\tau,\tau) \, \delta U([v;w], \tau,0) \psi_0
\end{align*}
The conditions of linearity and continuity needed for Fréchet differentiability can be directly transferred from Theorem \ref{Psi-frechet-prepare} as we add only finitely many terms.
\end{proof}

\begin{proof}[Proof of Theorem \ref{Psi-frechet}, estimate for functional variations of Schrödinger dynamics]
We start with the definition of the Fréchet derivative using the $H[v]$-interaction picture like in \eqref{gateaux-deriv-int} and by applying Minkowski's inequality. The transformation with the evolution operator $U([v],t,0)$ does not affect the $L^2$-norm, so we have $\|\delta\psi[v;w]\|_{2,\infty} = \|\delta\hat\psi[v;w]\|_{2,\infty}$.
\begin{align*}
\|\delta\psi[v;w]\|_{2,\infty} &= \sup_{t\in[0,T]} \left\| \int_0^t \hat{w}(s) \psi_0 \d s\right\|_2 \\
&\leq \int_0^T \|\hat{w}(s) \psi_0\|_2 \d s = \|\hat{w} \psi_0\|_{2,1}
\end{align*}
Next we apply the topological duality of $L^{2,\infty}-L^{2,1}$ with the time-space scalar product $\llangle\cdot,\cdot\rrangle$ to saturate the Hölder inequality with a special $\varphi \in L^{2,\infty} \subset \X$.
\begin{equation}\label{estimate-duality-1}
|\llangle\varphi,\hat w \psi_0\rrangle| = \|\varphi\|_{2,\infty} \cdot \|\hat w\psi_0\|_{2,1}
\end{equation}
Similarly we get by $\X-\X'$ duality and Hölder's inequality after substituting back the transformed $\hat w$ and moving one $U[v]$ to the left side of the scalar product
\begin{align}\label{estimate-duality-2}
\begin{split}
|\llangle\varphi,\hat w \psi_0\rrangle| &= |\llangle U[v]\varphi,w \psi[v]\rrangle | \\
&\leq \|U[v]\varphi\|_\X \cdot \|w\psi[v]\|_{\X'}.
\end{split}
\end{align}
Our aim will be to get an estimate for the r.h.s.~of \eqref{estimate-duality-2} which in return yields an inequality for $\|\delta\psi[v;w]\|_{2,\infty}$ over \eqref{estimate-duality-1}.
First we consider the term $\|U[v]\varphi\|_\X$ which has to be treated carefully, because it involves the time-dependent evolution of an also time-dependent trajectory, i.e. $t \mapsto U([v],t,0)\varphi(t)$. But we easily have $\|U([v],t,0)\varphi(t)\|_q \leq \sup_{s \in [0,T]} \|U([v],t,0)\varphi(s)\|_q$ and thus
\begin{align*}
\|U[v]\varphi\|_\X &= \|\varphi\|_{2,\infty} + \|U[v]\varphi\|_{q,\theta} \\
&\leq \|\varphi\|_{2,\infty} + \sup_{s \in [0,T]} \|U[v]\varphi(s)\|_{q,\theta}.
\end{align*}
The Strichartz estimate from Theorem \ref{Cv-strichartz} gives us $\|U[v]\varphi(s)\|_{q,\theta} \leq C_v \|\varphi(s)\|_2$ and we have in combination
\begin{align}\label{estimate-strichartz-3}
\begin{split}
\|U[v]\varphi\|_\X &\leq \|\varphi\|_{2,\infty} + C_v \sup\nolimits_{s \in [0,T]} \|\varphi(s)\|_2 \\
&= (1+C_v) \|\varphi\|_{2,\infty}.
\end{split}
\end{align}
The final term is $\|w\psi[v]\|_{\X'}$ from \eqref{estimate-duality-2} which is treated with Lemma \ref{lemma-mult-op} for estimating the action of the multiplication operator $w$ and then a second time with the Strichartz inequality from Theorem \ref{Cv-strichartz}.
\begin{align}\label{estimate-strichartz-4}
\begin{split}
\|w\psi[v]\|_{\X'} &\leq T^* \|w\|_\V \|\psi[v]\|_\X \\
&\leq T^* \|w\|_\V \cdot (1+C_v)\|\psi_0\|_2
\end{split}
\end{align}
We are now able to put \eqref{estimate-duality-1} and \eqref{estimate-duality-2} together with the estimates \eqref{estimate-strichartz-3} (where $\|\varphi\|_{2,\infty}$ cancels out) and \eqref{estimate-strichartz-4} above to state the inequality of the main theorem.
\end{proof}

\section{Fréchet differentiability and bounded observable quantities}

\label{FrechetBounded}

We now apply our main result and discuss consequences. We first provide an inequality that describes how strongly two solutions $\psi[v]$ and $\psi[v+w]$ differ. Since we have shown Fréchet differentiability of the wave function we can employ the fundamental theorem of calculus \cite[Cor.~30.1.1]{blanchard-bruening} and find
\begin{align*}
&\psi([v+w],t) - \psi([v],t) \\
&= - \i \int_{0}^{1} \left( \int_0^{t} U([v_{\lambda}],t,s) \, w(s) \, \psi([v_{\lambda}],s) \d s\right) \d \lambda,
\end{align*}
where $v_{\lambda}=v+\lambda w$. Using the inequality in Theorem \ref{Psi-frechet} we then find that
\begin{equation}\label{diff-trajectories}
 \|  \psi[v+w] - \psi[v] \|_{2,\infty} \leq C \, T^* \|w\|_\V \|\psi_0\|_2,
\end{equation}
where $C = \sup_{\lambda \in [0,1]}(1 + C_{v_\lambda})^2$. Hence if the external potentials and interactions measured in the norm of $\V$ only differ slightly, the resulting trajectories are very similar with respect to the norm of $\mathcal{C}^{0}([0,T], \H)$. Consequently the Schr\"odinger dynamics is stable with respect to small changes or inaccuracies in the external potentials and interactions. The same applies to bounded operators. For this, consider the expectation value of a time-independent, self-adjoint, bounded operator $A : \H \rightarrow \H$ for a fixed initial state $\psi_0$ at time $t \in [0,T]$,
\[
\langle A \rangle_{[v]}(t) = \langle \psi([v],t) ,A \psi([v],t) \rangle.
\]

Using the product rule for functional variations of potentials and switching to the $H[v]$-interaction picture once more we get the following from \eqref{gateaux-deriv-int} and $\hat\psi([v],t) = \psi_0$. (Note: The scalar product is antilinear in the first component; ``\textit{c.c.}" stands for the complex conjugate of the whole expression.)
\begin{align}\label{var-A}
\begin{split}
\delta \langle A \rangle_{[v;w]}(t) &= \langle \delta \psi([v;w],t), A \psi([v],t) \rangle + c.c. \\[0.5em]
&= \langle \delta \hat\psi([v;w],t), \hat A(t) \hat\psi([v],t) \rangle + c.c. \\
&= \i \int_0^t \langle \hat w(s) \psi_0, \hat A(t) \psi_0 \rangle \d s + c.c. \\
&= \i\int_0^t \langle [\hat{w}(s), \hat{A}(t)] \rangle_0 \d s 
\end{split}
\end{align}
This is exactly the Kubo formula of first order perturbations central to linear-response theory. Note especially that $\hat A(t)$ gets time-dependent because of the $H[v]$-interaction picture transformation with $U([v],t,0)$.

\section{Fréchet differentiability and time-dependent density-functional theory}

\label{FrechetDensity}

Another important quantity though not a self-adjoint operator is the one-particle density. We adopt the notation $x=x_1, \bar{x} = (x_2,\ldots,x_N)$. For spatially (anti-)symmetric trajectories $\psi[v] \in \Cont^0([0,T], \H) \supset \X$ the density is defined as
\[
n([v],t,x) = N \smashoperator{\int\limits_{\R^{3(N-1)}}} \d \bar{x}\, |\psi([v],t,x,\bar{x})|^2.
\]
Within our framework it is now natural to ask for the Fréchet derivative $\delta n[v;w]$. Like in \eqref{var-A} we get
\begin{align*}
&\delta n([v;w],t,x) \\
&= N \smashoperator{\int\limits_{\R^{3(N-1)}}} \d \bar{x}\, \overline\psi([v],t,x,\bar{x}) \delta\psi([v;w],t,x,\bar{x}) + c.c.
\end{align*}
An estimate can now easily be established with Theorem \ref{Psi-frechet}.
\begin{align*}
\sup_{t \in [0,T]}& \|\delta n([v;w],t)\|_1 = \sup_{t \in [0,T]} \int_{\R^3} \d x \left| \delta n([v;w],t,x) \right| \\
&\leq 2 N \sup_{t \in [0,T]} \langle |\psi([v],t)|, |\delta\psi([v;w],t)| \rangle \\
&\leq 2 N \sup_{t \in [0,T]} \|\psi([v],t)\|_2 \cdot \|\delta\psi([v;w],t)\|_2 \\
&\leq 2 N (1+C_{v})^2 T^* \|w\|_\V \cdot \|\psi_0\|_2^2
\end{align*}

To make the connection to physics and standard density-response theory \cite{stefanucci-vanleeuwen} more explicit we further restrict ourselves to (anti-)symmetric trajectories associated with spatially symmetric $v \in \V$ and consider only symmetric (one-body) perturbations of the form $\sum_{k=1}^{N}w(t,x_k)  \in \V $. Furthermore we adopt the usual tacit assumption that the unitary evolution operator $U([v],t,s)$ can be represented by an integral transformation with an integral kernel (the so-called propagator) of the form $U([v],t,x,\bar{x},s,y,\bar{y})$. Then the above functional derivative can be rewritten as
\[
 \delta n([v;w],t,x) = \int_{0}^{t} \d s \int_{\mathbb{R}^3} \d y \, \chi([v],t,x,s,y) \, w(s,y),
\]
where the non-equilibrium linear-response kernel is defined by 
\begin{align*}
 \chi([v],t,x,s,y) &= -\i N^2 \smashoperator{\int\limits_{\R^{6(N-1)}}} \d \bar{x}  \d \bar{y}\, \overline{\psi}([v],t,x,\bar{x})\cdot \\
 &\cdot U([v],t,x,\bar{x},s,y,\bar{y}) \psi([v],s,y,\bar{y}) + c.c.
\end{align*}

Finally, we discuss these results in the context of TDDFT. In this many-body theory the basic statements are concerned with the bijectivity of the mapping $v \mapsto n$ for a fixed initial state and given two-body interaction \cite{TDDFT, TDDFT2}. The allowed potentials take the symmetric (one-body) form $v \equiv \sum_{k=1}^{N}v(t,x_k)$ and potentials that only differ by a spatially constant function $c(t)$ are considered equivalent since their action only amounts to a change of gauge and does not influence the dynamics of the quantum system.
Thus the set of allowed potentials consists of equivalence classes of one-body potentials $v$. The standard approach to show bijectivity of such a mapping $v \mapsto n$ is due to Runge and Gross \cite{runge}, which restricts the set of potentials further to those that are Taylor-expandable in time. Several approaches have been developed in the recent years that try to overcome this restriction \cite{vanleeuwen2, ruggenthaler1, ruggenthaler3, tokatly, farzanehpour}. For example the local invertibility of the mapping $v \mapsto n$ is considered for Laplace-transformable potentials that perturb a (time-independent) many-body system in its ground state\cite{vanleeuwen2}. The main ingredient in that approach is the linear-response kernel in Lehmann representation. For explicitly time-dependent problems, however, the Lehmann representation is no longer valid, and thus such local invertibility considerations need to be based on non-equilibrium density-response theory. The current result thus sets the stage for similar considerations based on the inverse-function theorem for Banach spaces (which needs Fréchet differentiability) also in the case of time-dependent systems. Additionally a fixed-point approach to TDDFT was developed\cite{ruggenthaler1, ruggenthaler2}, where one employs the Fréchet differentiability of the divergence of the internal local-force density $q[v]$ which can be formally defined by \cite{TDDFT}
\[
 q([v],t,x) = \partial_t^2 n([v], t,x) - \nabla \cdot\left[ n([v],t,x) \nabla v(t,x) \right]. 
\]

One immediately sees that $q[v]$ is only well-defined if the density $n[v]$ and the potential $v$ obey certain differentiability conditions in space and time. Alternatively one can define $q[v]$ directly in terms of the wave function (see Ullrich \cite[3.49]{TDDFT}), from which we see that a sufficient condition for the existence of $q[v]$ is that the wave function is four-times differentiable in space. Either way, we would need to impose further regularity conditions on the solutions of the TDSE and hence consider Fréchet differentiability with respect to stronger norms on the space of potentials and trajectories.
At the moment strict though superfluously strong regularity conditions for solutions to the TDSE are known for very specific situations such as periodic systems with infinitely-differentiable potentials with respect to space and time \cite{delort}. Hence, while we cannot yet give a full proof for the differentiability of $q[v]$, our results show that the assumption of Fréchet differentiability is mathematically reasonable, and provide a further step towards a rigorous fixed-point formulation of TDDFT. For further details on rigorous results in TDDFT and their connection with Fréchet differentiability we refer to a recent review article \cite{review}.

\section{Conclusion and outlook}

In this work we have shown Fréchet differentiability of the many-body wave function with respect to time-dependent potentials. This implies that the difference of two solutions of the TDSE that start from the same initial state is bounded by the difference of the external potentials and interactions. Hence the dynamics of multi-particle systems is not very sensitive to small perturbations in the external fields or interactions. The Fréchet differentiability can be directly extended to bounded operators and the one-particle density, which leads to the non-equilibrium version of the Kubo formula and non-equilibrium density-response theory. This is important since the usual Lehmann representation of the linear-response kernel is not valid in explicitly time-dependent systems. Further, the current result sets the stage for local-inversion investigations of the mapping $v \mapsto n$ for explicitly time-dependent problems and shows that the assumption of a Fréchet differentiable internal-force density $q[v]$ in the fixed-point formulation is reasonable. However, a rigorous proof for the differentiability of $q[v]$ is still open.

Since the current approach effectively rules out Coulombic potentials for more than one particle in three dimensions it would be desirable to extend the results of this work to the approach of Reed and Simon \cite[Th.~X.71]{reed-simon} which can also treat this important case. After all, in the quantum-mechanical modelling of many-body systems usually Coulombic potentials are employed. However, we expect that the wave functions stay Fréchet differentiable in such situations as well, and hence the dynamics of many-particle systems in three dimensions do not depend drastically on whether one uses the exact or an approximate Coulombic potential. Further, to rigorously investigate the functional differentiability of $q[v]$ explicit conditions on the potentials to guarantee four-times differentiable wave functions need to be devised. These two related questions will be the subject of future work and will allow us to further strengthen the foundations of TDDFT.

\begin{acknowledgements}
Fruitful discussions regarding this work were held with Robert van Leeuwen, University of Jyväskylä and Klaas Giesbertz, VU University Amsterdam. Special thanks go to Neepa T. Maitra and her group for the great hospitality we experienced during our stay at Hunter College, New York City University. M.R. acknowledges financial support by the Austrian Science Fund (FWF projects J 3016-N16 and P 25739-N27).
\end{acknowledgements}

\end{document}